\documentclass[11pt]{article}
\usepackage{amsmath,amssymb,amsbsy,amsfonts, amsopn,amstext,amsxtra,euscript,amscd,graphicx}
\usepackage[paper=a4paper, left=2.5cm, right=2.5cm, top=2cm, bottom=2cm,includefoot, centering]{geometry}
\usepackage{textcomp,enumerate,float,bm,setspace,boxedminipage,color,rotating}
\usepackage{hyperref}
\usepackage[english]{babel}
\usepackage{diagbox}

\newtheorem{definition}{Definition}

\newtheorem{theorem}{Theorem}
\newtheorem{example}{Example}

\newenvironment{proof}{\vspace{8pt}
\noindent{\bf Proof}: }{{\hfill {\large $\Box$}} \vspace{8pt}}

\usepackage{authblk,xcolor}
\usepackage{extarrows}
\usepackage{bm}

\newcommand{\remove}[1]{}

\usepackage[numbers,sort&compress]{natbib}

\begin{document}

  \title{
  A Unified Framework for Constructing Information-Theoretic Private Information Retrieval}

\author{Liang Feng Zhang}

\affil{\small School of Information Science and Technology, \\ ShanghaiTech University, Shanghai, China
\\ zhanglf@shanghaitech.edu.cn}

\date{}

\maketitle

\begin{abstract}
 Retrieving up-to-date information from a publicly accessible database poses significant threats to
the user's privacy. {\em Private information retrieval} (PIR) protocols 
allow a user to retrieve any entry from a database, without revealing the identity of the entry being retrieved
to the server(s). 
Such protocols have found numerous applications in both
 theoretical studies and  real-life scenarios. 
The existing PIR constructions mainly give  
multi-server {\em information-theoretic} PIR (IT-PIR) protocols  
or  single-server computational PIR (CPIR)  protocols. 
Compared with CPIR, IT-PIR protocols are computationally more efficient and 
secure in the presence of   unbounded servers.
The  most classical and challenging problem in the realm of IT-PIR  is
constructing protocols with lower {\em communication complexity}. 
In this review, we introduce a new discrete structure called 
{\em families of
orthogonal arrays with span capability} (FOASC) and propose a unified framework
for constructing IT-PIR protocols. 
We show how the most influential IT-PIR protocols in the 
literature can be captured by the framework. We also put forward several interesting
open problems concerning FOASC, whose solutions may 
result in innovative IT-PIR protocols.

\vspace{2mm}
\end{abstract}

\section{Introduction}

Publicly accessible databases are indispensable resources for retrieving up-to-date information. Access to such databases 
 poses significant risks to the privacy of the user, since the database server(s)
 may monitor the user's queries and infer what the user is after. 
Usually the user's retrieval  intent is  highly valuable 
 and  needs careful protection. For example, 
 for a stock-market  database an   investor's  retrieval intent may 
  influence the stock's price; for a patent database a company's retrieval intent may
attract  unexpected  pursuer of the patent; for a Merkle proof database on which a 
blockchain system such as Ethereum is based, a user's retrieval intent may 
link the user to the account being read and eventually lead to deanonymization. 
  
Private information retrieval (PIR) protocols \cite{CGKS95} are cryptographic protocols that are specifically designed to  ensure the users' privacy. 
Such protocols allow a user to retrieve an entry
$x_i$ from a database ${\bm x}=x_1\cdots x_n\in \{0,1\}^n$, without revealing 
the retrieval index   $i\in[n]$ to 
the server.
At first glance, the  requirements posed by PIR  seem 
quite absurd. 
However, there does exist a trivial solution that {\em perfectly}
hides $i$ from the server,  where the user simply downloads the entire 
database ${\bm x}$ from the server and then locally extracts  $x_i$.  
In particular, the perfect privacy  is {\em information-theoretic} and   means that the server learns absolutely no
information about $i$, even if it has unlimited computing power. 
This trivial solution incurs a {\em communication cost} of $O(n)$, which 
could be prohibitive if the database consists of millions or billions of entries.  
Unfortunately, in their pioneer work \cite{CGKS95}, Chor, Goldreich, Kushilevitz and Sudan
showed that the $O(n)$ communication cost of the trivial solution is asymptotically optimal,
 if there is only {\em one server} and {\em perfect privacy}   is needed. 
 Therefore, to have a PIR solution of communication cost $o(n)$, the user must consider two possible
 relaxations: (1)    resort to multiple  servers; (2)  give up the perfect privacy.

 \vspace{2mm}
\noindent
{\bf Two flavors of PIR.}
   Under the first relaxation, the user may 
 communicate with $k~ (k>1)$ servers,  send 
 a query to every server, receive an answer from the server, and 
 finally reconstruct   $x_i$ from the $k$ answers. Specifically, 
 each of the servers should store a copy of the same database $\bm x$ and answer the user's query 
 with  $\bm x$. To differentiate from single-server solutions, the 
 $k$ servers must not collude with each other. If the user's retrieval index $i$ is perfectly  (i.e.,
 information-theoretically)  hidden
from  the collusion of any $t~(t<k)$ out of the $k$ servers,  then the protocol is said to be a 
 $t$-{\em private} $k$-server   information-theoretic   PIR (IT-PIR) \cite{CGKS95}, or {\em $(t,k)$-PIR} for short. 
Under the second relaxation, the user may properly encode its retrieval index 
$i$ as a query, which essentially leaks no information about 
$i$ to   any {\em computationally bounded} server that runs polynomial-time algorithms, such that the server remains able to
compute an encoded form of $x_i$ to the user. In particular, the 
privacy of $i$ must be built on various
number-theoretic problems (e.g., the quadratic residuosity problem, the composite residuosity 
problem), which are hard to solve in feasible time by
the computationally bounded server. 
Protocols in this category have been called single-server  computational PIR    \cite{KO97}, or {\em CPIR} for short.

 \vspace{2mm}
\noindent
{\bf Practical influence of PIR.}
 Both CPIR and IT-PIR 
are  important cryptographic primitives that have practical influences. Today  
 PIR protocols have 
  found numerous applications in real-life scenarios, e.g.,  
  private database search \cite{WGY+17}, metadata hiding messaging \cite{AS16,ACL+18}, private media consumption \cite{GCM+16},
 private contact
discovery \cite{KRS19},  private blocklist lookups \cite{KG21},  privacy-friendly advertising \cite{BKM+12,SHD21}, 
 certificate transparency \cite{HHG+23}, 
private web search \cite{HDG+23}, private electronic commerce \cite{HOG11},
and private location based services \cite{GKK+08}, among others. 
Recently, commercial systems such as 
Microsoft's Password Monitor \cite{CHL+18},  Google's  Device Enrollment 
\cite{Google}, Blyss's  Spiral \cite{MW22}, and 
Brave's  FrodoPIR \cite{DPC23} have integrated the functionality of PIR and 
signed the real world deployment of PIR.

\vspace{2mm}
\noindent
{\bf Theoretical influence of PIR.}
On the theoretical side, both IT-PIR and CPIR 
are fundamental  building blocks of many other cryptographic primitives and have their featured applications. 
IT-PIR protocols may give locally decodable codes (LDCs)
\cite{KT00,Yek12,Z14},  error-correcting codes 
that can recover any bit of the message by reading  a few bits  of the codeword  and guarantee 
 correct recovery even if a constant fraction of 
the codeword have been {\em adversarially} corrupted.  IT-PIR protocols can also be used to 
construct multi-party information-theoretically private protocols \cite{IK04,BCWZ12}. 
CPIR protocols imply many important 
cryptographic primitives such as unconditionally hiding
commitment \cite{BIKM99}, oblivious transfer \cite{CMO00,NP99},   collision-resistant hash functions \cite{IKO05}, 
and
 efficient zero-knowledge arguments \cite{KR06}.

\vspace{2mm}
\noindent
{\bf Communication cost.}
The efficiency of PIR protocols is mainly measured by
   {\em communication complexity} \cite{CGKS95},     the total number of bits that have to be exchanged between the user and the server(s) in order to retrieve one bit from the database.  
The most classical and challenging problem in the realm of 
PIR is constructing  protocols with lower  
  communication complexity for a given number of servers. 
  While there are $O(\log n)$-server PIR protocols with polylogrithmic 
  (in $n$) communication complexity, 
the main focus has been  protocols that use a {\em constant}  number of servers. 
For IT-PIR, after a long line of
arduous explorations  \cite{CGKS95,Amb97,Ito99,IK99,BI01,WY05,Yek07,Rag07,KY08,Efr09,IS10,CFLWZ13,DG15,ABL24,GKS25}, 
 today the most efficient protocols that use a constant number of servers    have reached 
 a  communication complexity that is subpolynomial in $n$. 
For CPIR, protocols   \cite{KO97,CMS99,Cha04,Lip05,GR05,MG07,Lip09,GKL10,YKPB13,DC14,KLLPT15,LP17} based on various 
cryptographic assumptions have been proposed and  the up-to-date ones 
may achieve an optimal rate that is close to 1.

\vspace{2mm}
\noindent
{\bf Computation cost.}
Beimel,   Ishai and Malkin \cite{BIM00} showed that in any PIR protocol 
 every entry of the database $\bm x$ must be accessed 
at least once by the servers,  in order for the user's retrieval index to be private. 
The observation is reasonable because any non-accessed entry $x_j$ cannot be of the user's interest 
and thus reveals  partial information about the user's retrieval index $i$ (i.e., $i\neq j$) 
to the server(s).  Consequently, in any PIR protocol the servers computation cost must be 
  $\Omega(n)$, which could be rather undesirable for a large $n$. 
  In particular, for IT-PIR  the servers may 
  need to perform $\Omega(n)$ field operations; for CPIR,
  the $\Omega(n)$ operations could be expensive public-key operations such as 
  exponentiations.  
Sion and  Carbunar \cite{SC07} even concluded 
that   deployment
of non-trivial CPIR protocols on real hardware
 would be orders of magnitude less
time-efficient than trivially transferring the entire database. 
Starting from \cite{BIM00}, there have been a long line of 
research that tried to obtain computationally efficient   
  IT-PIR  \cite{BIM00,WY05,SWZ24,ISW24} and 
CPIR \cite{BIPW17,CHR17,CK20,GZS24,LMW23,RMS24,PPY18,SACM21,ZLTS23,ZPZS24,ACL+18,ALP+21,MCR21,MK22,GHK22,MW22,DPC23,HHG+23,ZLTS23,LP23,LLW24,LMRS24,
LP24,FLLP24,MW24,HPPY25,WR25} protocols.

\vspace{2mm}
\noindent
{\bf PIR against malicious servers.}
While most of the existing PIR protocols assume  {\em honest-but-curious}  servers
that always faithfully execute the protocol,   {\em malicious} servers
may arbitrarily deviate from the predefined specifications
and thus prevent the correct
execution of the protocol. 
In particular, the malicious servers may not respond to the user's
queries or even  tamper with the   responses. 
Such behaviors may lead to failure in retrieval. 
 Beimel and Stahl   \cite{BS02}
initiated the study of {\em  robust} $k$ out of $\ell$ PIR   protocols that 
allow the user to contact $\ell$ servers and successfully retrieve
$x_i$ as long as at least  $k$ out of the $\ell$ servers respond, and
 $b$-{\em Byzantine} robust $k$ out of $\ell$ PIR   protocols
 that still guarantee successful retrieval even if 
 $b$ out of the $k$ responses are tampered with.
 For $k=\ell$,  today such protocols are also termed as 
 $b$-{\em error correcting} $k$-server PIR protocols  
   \cite{BS02,Gol07,DGH12,Kur19,ZWW22,EKN22,EKN24}.
   Zhang and Safavi-Naini \cite{ZS14}
   initiated the study of $b$-{\em error detecting} $k$-server PIR protocols 
   that can detect  the existence of wrong responses. 
   Such protocols     \cite{ZLL22,CTD23,ZW22,KZ22,KZ23,KDK24,EKN221,CNG23,DT24} are particular useful when the
   PIR servers are implemented by untrusted  cloud servers.

\vspace{2mm}
\noindent
{\bf IT-PIR vs. CPIR.}
Compared with IT-PIR protocols,  
CPIR protocols do not require the user to communicate with
multiple non-colluding servers, an arguably strong assumption. Furthermore, 
they may achieve much lower communication complexity, compared with 
constant-server IT-PIR.
On the negative side, CPIR protocols are computationally extensive and cannot 
have short queries or responses, which are 
crucial for constructing LDCs. Also, the cryptographic assumptions underlying CPIR may 
become fragile in the presence of modern computing technologies, which however 
cannot affect the security of IT-PIR.  
In this review, we are restricted to IT-PIR and focus on the long line of works on
constructing 
communication efficient protocols, which has been the most challenging research problem.

  \vspace{2mm}
\noindent
{\bf Related work.}
Our review focuses on a unified framework for constructing IT-PIR protocols in the honest-but-curious server model and is
 different from several excellent existing   reviews, which either 
cover IT-PIR constructions before 2007  and provide  no
unified framework \cite{Gas04,Bei08,Blu11} or  focus  on CPIR \cite{OS07}.  

   \vspace{2mm}
\noindent
{\sc Organization.}
In Section \ref{sec:pre} we give  the definitions of IT-PIR and orthogonal arrays;
In Section \ref{sec:frwk} we propose the notion of family of OAs with  span capability (FOASC) and give an 
FOASC based  framework for constructing IT-PIR;
in Section \ref{sec:pir} we show how several most influential IT-PIR constructions can be captured by
the proposed framework; in Section \ref{sec:open} we discuss several open problems concerning FOASC. 
Finally, Section \ref{sec:con}   concludes the review.

\section{Preliminaries}
\label{sec:pre}

{\bf Notation.} For any integer $n>0$, we denote  $[n]=\{1,\ldots,n\}$. 
For any  prime power  $p$, we denote by $\mathbb{F}_p$ the {\em finite field} of $p$ elements. 
For any two vectors $\bm u=(u_1,\ldots,u_m)$ and $\bm v=(v_1,\ldots,v_n)$,
we denote by    $\bm u\|\bm v=(u_1,\ldots,u_m,v_1,\ldots,v_n)$ the {\em concatenation} of $\bm u$ and $\bm v$.
For any $m\times n$ matrix $\bm Q$, we denote by $\bm Q^\top$ the {\em transpose} of $\bm Q$ and denote by
$Q_{i,j}$ the {\em $(i,j)$-entry} of $\bm Q$ for all $i\in[m]$ and $j\in[n]$.   
For any integers $h>0$ and $i\in[h]$, we denote by $\bm e_h^{(i)}$ the length-$h$ {\em unit vector}
whose $i$th entry is 1 and all other entries are 0.
For any predicate $P$, we denote by $1_P$  the indicator value for $P$, i.e., $1_P=1$ if $P$ is true and
0 otherwise. For example, $1_{3\in [2]}=0$. 
For any vectors $\bm z=(z_1,\ldots,z_n)$ and $\bm u=(u_1,\ldots,u_n)$, we denote 
${\bm z}^{\bm u}=(z_1)^{u_1}\cdots (z_n)^{u_n}$. 

\subsection{Private Information Retrieval}

A $t$-private $k$-server PIR  ($(t,k)$-PIR) protocol  involves $k$ {\em servers} ${\cal S}_1, \ldots, {\cal S}_k$, each storing  a copy of the database 
  $\bm x=(x_1,\ldots, x_n)\in\{0,1\}^n$, and a {\em user} $\cal U$ who  wants to retrieve a database entry
$x_i$, without revealing the retrieval index $ i \in [n]$ to any $t$ out of the $k$ servers. 
 \begin{definition}
[Private Information Retrieval]
 A {\em $t$-private $k$-server private information retrieval ($(t,k)$-PIR)} protocol ${\cal P= (Q,A, C)}$
  consists of three algorithms as follows:
 \begin{itemize}
 \item $(q_1,\ldots,q_k,{\sf aux})\leftarrow {\cal Q}(k,n,i)$:  a randomized {\em querying} algorithm
 that takes the public parameters $k,n$ and the user's private {\em retrieval index} $i\in[n]$
 as input, and outputs $k$ {\em queries} $q_1,\ldots,q_k$ together with an 
 {\em auxiliary information string}
 $\sf aux$ for reconstruction.

\item $a_j\leftarrow {\cal A}(k, j, \bm x, q_j)$:
a deterministic {\em answering} algorithm
that takes the database $\bm x=(x_1,\ldots, x_n)\in \{0,1\}^n$ and the query $q_j$
as input, and outputs an {\em answer} $a_j $.

\item $x_i\leftarrow {\cal C}(k, n, a_1, \ldots, a_k, {\sf aux})$:
a deterministic {\em reconstructing} algorithm that takes $\sf aux$ and
the  $k$ answers $a_1,\ldots,a_k$   as input, and outputs the  {\em target  entry}  $x_i$. 
 
\end{itemize}
For a  protocol as above to be a $(t,k)$-PIR, the following   requirements should be satisfied:
\begin{itemize}
\item {\bf Correctness}. 
Informally,   if all algorithms of the protocol $\cal P$ are 
faithfully executed, then the reconstructing algorithm
always  outputs the correct value of the target entry. Formally,  for  any $\bm x \in \{0, 1\}^n$,
$i \in [n]$, $(q_1,\ldots,q_k,{\sf aux})\leftarrow {\cal Q}(k,n,i)$, and   
 $\{a_j\leftarrow {\cal A}(k, j, \bm x, q_j)\}_{j=1}^k$,  
$${\cal C}(k, n, a_1, \ldots, a_k, {\sf aux})=x_i.$$

\item {\bf $t$-Privacy}. Informally,  
any   collusion of   $\leq t$ servers learns no information about the user's retrieval index $i$. 
Formally,   for any $i_1, i_2 \in [n]$,  
 any subset $T \subseteq [k]$ of size $\leq t$,  
$${\cal Q}_T (k, n, i_1)\xlongequal{\rm id} {\cal Q}_T (k, n, i_2),$$ where ${\cal Q}_T$ denotes
concatenation of $j$-th outputs of $\cal Q$ for all  $j\in   T$ and    ``$\xlongequal{\rm id}$"
means that two distributions are identical.
\end{itemize}

\end{definition}

\input{model.tpx}

\newpage
\vspace{2mm}
\noindent
{\bf PIR System.}
In a $(t,k)$-PIR system (Figure \ref{fig:pir}), the user   $\cal U$ starts the execution of the protocol  by invoking
 ${\cal Q}(k, n, i)$ to pick a random 
$k$-tuple of queries $(q_1, \ldots, q_k)$ along with an auxiliary information string $\sf aux$, and then 
sending   each query $q_j$ to the server ${\cal S}_j$.
 Subsequently, each server  ${\cal S}_j$ invokes the answering algorithm 
 ${\cal A}(k, j, \bm x, q_j)$ to compute an answer $a_j$ to the user. 
 Finally, $\cal U$ reconstructs  $x_i$ by executing  the reconstructing  algorithm
${\cal C}(k, n, a_1, \ldots, a_k, {\sf aux})$.

\vspace{2mm}
\noindent
{\bf Communication Complexity.}
The {\em communication complexity} of a PIR protocol ${\cal P}$, denoted by ${\bf C}_{\cal P}(n, k)$, is a function
of $k$ and $n$ that measures  the total number of bits communicated between the user
and $k$ servers, maximized over all choices of the database $\bm x=(x_1,\ldots, x_n) \in \{0, 1\}^n$, the retrieval index $i\in [n]$, and the random coins of the querying algorithm $\cal Q$.

\subsection{Orthogonal Arrays}

 Orthogonal arrays (OAs) \cite{HSS99}  have played  prominent roles 
 in the design of experiments and  found many applications in  computer science.
In this review, we shall use OAs to give a unified framework for 
IT-PIR. 
For any integers $N,k,t>0$, we use the term {\em ``$N\times k$ array"}
to refer to a matrix $\bm Q$ with $N$ rows and $k$ columns, and
use the term {\em ``$N\times t$ subarray"} to refer to a  submatrix of 
$\bm Q$ that consists of   $t$ columns of $\bm Q$, where $t\leq k$.   

  \begin{definition}
 [Orthogonal Array] Let $N,k,s,t>0$ be integers.  Let $\mathbb{S}$ be a set of $s$ symbols or levels. An
  $N\times k$ array $\bm Q$ is said to be an {\em orthogonal array (OA)} of  
 {\em level} $s$, {\em strength} $t$, and {\em index} $\lambda$, or ${\sf OA}(N,k,s,t)$ for short, 
  if every $N\times t$ subarray of $\bm Q$ contains every element of 
  $\mathbb{S}^t$ exactly $\lambda$ times as a row. 
\end{definition}

\begin{example}
The following $8\times 4$ array is an  ${\sf OA}(8,4,2,3)$  with index 1 (where $\mathbb{S}=\{0,1\}$): 
$$
  \begin{array}{cccc}
    0 & 0 & 0 & 0 \\
    0 & 0 & 1 & 1 \\
    0 & 1 & 0 & 1 \\
    0 & 1 & 1 & 0 \\
    1 & 0 & 0 & 1 \\
    1 & 0 & 1 & 0 \\
    1 & 1 & 0 & 0 \\
    1 & 1 & 1 & 1 \\
  \end{array}
$$ 
\end{example}

\section{A   Framework based on Families of  Orthogonal Arrays}
\label{sec:frwk}

In this section, we propose a unified framework that  captures several of the most influential  constructions of
IT-PIR protocols 
\cite{CGKS95,IK99,BI01,WY05,Yek07,Efr09,CFLWZ13,DG15,GKS25} during the past 30 years.

\subsection{Families of Orthogonal Arrays with Span Capability}

Our framework is based on a new discrete structure called 
family of orthogonal arrays with  span capability (FOASC),
 which is a set of OAs that satisfy special algebraic properties.  
\begin{definition}
[Families of Orthogonal Arrays with Span Capability] Let $N,k,s,t,n>0$ be   integers. 
Let $\mathbb{S}$ be a set of $s$  levels and let
$\mathbb{R}$ be a commutative ring with identity. 
Let $\bm \alpha=(\alpha_1,\ldots,\alpha_n)$ be $n$ functions with domain $\mathbb{S}$ and range   $\mathbb{R}$. 
We say that a set $\{\bm Q^{(1)},\ldots, \bm Q^{(n)}\}$ of  ${\sf OA}(N,k,s,t)$'s 
is a {\em family of   orthogonal arrays with $\bm \alpha$-span capability}, or 
${\sf FOASC}(N,k,s,t;\bm \alpha)$ for short, 
if for all $i\in[n]$ and   $\ell\in [N]$, the columns of the following matrix   
$$
{\bm \alpha}(\bm Q^{(i)}_\ell)=
\left(
  \begin{array}{cccc}
\alpha_1(Q^{(i)}_{\ell,1}) & \alpha_1(Q^{(i)}_{\ell,2}) & \cdots & \alpha_1(Q^{(i)}_{\ell,k}) \\
\alpha_2(Q^{(i)}_{\ell,1}) & \alpha_2(Q^{(i)}_{\ell,2}) & \cdots & \alpha_2(Q^{(i)}_{\ell,k}) \\
\vdots & \vdots & \cdots & \vdots \\
 \alpha_n(Q^{(i)}_{\ell,1}) & \alpha_n(Q^{(i)}_{\ell,2}) & \cdots & \alpha_n(Q^{(i)}_{\ell,k}) \\
  \end{array}
\right)
$$
span a nonzero multiple of  $\bm e^{(i)}_n$,
where $\bm Q^{(i)}_\ell$ stands for the $\ell$th row of $\bm Q^{(i)}$. 
\end{definition}

\begin{example}
The following OAs $\bm Q^{(1)}, \bm Q^{(2)}$ form an  ${\sf FOASC}(9,2,9,1;\bm \alpha)$,
where $\bm \alpha=(\alpha_1,\alpha_2)$ and $\alpha_1,\alpha_2$ are functions with domain $\mathbb{S}=\mathbb{F}_3^2$ and range $\mathbb{R}=\mathbb{F}_3$ such that 
$\alpha_1(a,b)=a, \alpha_2(a,b)=b$.
$$
\underbrace{\begin{array}{cc}
 \text{\bf col 1}  & \text{\bf col 2} \\
(1, 0) & (1, 0)  \\    
(1, 1) & (1, 2)  \\    
(1, 2) & (1, 1)  \\    
(2, 0) & (0, 0)  \\    
(2, 1) & (0, 2)  \\    
(2, 2) & (0, 1)  \\    
(0, 0) & (2, 0)  \\    
(0, 1) & (2, 2)  \\    
(0, 2) & (2, 1)        
\end{array}}_{\bm Q^{(1)}}
\hspace{2cm}
\underbrace{\begin{array}{cc}
 \text{\bf col 1}  & \text{\bf col 2} \\
(0, 1) & (0 , 1)  \\
(0, 2) & (0 , 0)  \\
(0, 0) & (0 , 2)  \\
(1, 1) & (2 , 1)  \\
(1, 2) & (2 , 0)  \\
(1, 0) & (2 , 2)  \\
(2, 1) & (1 , 1)  \\
(2, 2) & (1 , 0)  \\
(2, 0) & (1 , 2)  
\end{array}}_{\bm Q^{(2)}}
$$
In fact, there is a vector  $ \bm \lambda=(2,2)^\top$ such that $\bm \alpha(\bm Q^{(i)}_\ell)\cdot  \bm \lambda=
{\bm e}_2^{(i)}$ for all $i\in[2]$ and $\ell\in[9]$.
\end{example}

\subsection{The Framework}

In this section, we show a unified framework  (see Figure \ref{fig:frwk}) 
for constructing  $(t,k)$-PIR protocols from 
FOASC.  
Given an ${\sf FOASC}(N,k,s,t;  \bm \alpha)$ that consists of $n$ ${\sf OA}(N,k,s,t)$'s $\bm Q^{(1)},\ldots,\bm Q^{(n)}$, the main idea underlying 
our framework   is as follows: 
{\em   interpret  the database $\bm x\in \{0,1\}^n$ as a vector in
$\mathbb{R}^n$, encode the database $\bm x$ as a function $F_{\bm x}:\mathbb{S}\rightarrow \mathbb{R}$, which is essentially   a
 linear combination of the $n$ functions 
$\alpha_1,\ldots,\alpha_n$, i.e., 
\begin{align}\label{eq:Fx}
F_{\bm x}(\bm z)=\sum_{\tau=1}^n x_\tau\cdot \alpha_\tau (\bm z),
\end{align}
and finally reduce the problem of retrieving  $x_i$ to that  of 
evaluating   $F_{\bm x}$ on a random row of  ${\bm Q}^{(i)}$.}
\begin{figure}[H]
\begin{center}
\begin{boxedminipage}{16cm}

\vspace{1mm}
{\bf The underlying public parameters and structures:} 
\vspace{2mm}

\begin{itemize}
  \item $\bm Q^{(1)},\ldots,\bm Q^{(n)}$: an ${\sf FOASC}(N,k,s,t;\bm \alpha)$, where
 $\bm \alpha$ consists of   $n$ functions  $\alpha_1,\ldots,\alpha_n:\mathbb{S}\rightarrow \mathbb{R}$ from a set $\mathbb{S}$ of $s$ levels
 to a commutative ring $\mathbb{R}$ with identity. 
  \item $F_{\bm x}$:  a function representing  the database $\bm x$, based on
   the ${\sf FOASC}(N,k,s,t;\bm \alpha)$ (Eq. \eqref{eq:Fx}). 
 \item $\{\bm \lambda^{(i)}_\ell,\omega^{(i)}_\ell\}$:  a vector  $\bm\lambda^{(i)}_\ell=
(\lambda^{(i)}_{\ell,1},\ldots,
\lambda^{(i)}_{\ell,k})^\top\in
\mathbb{R}^k$ and a nonzero ring element $\omega^{(i)}_\ell\in \mathbb{R}$  such that
  $\bm \alpha(\bm Q^{(i)}_\ell)\cdot  \bm \lambda^{(i)}_\ell=
\omega^{(i)}_\ell {\bm e}_n^{(i)}$ for all  $i\in[n]$ and $\ell\in[N]$.
  \end{itemize}

\vspace{3mm}
{\bf The private information retrieval protocol ${\cal P}={\cal (Q,A,C)}$:}
\begin{itemize}
\item ${\cal Q}(k,n,i)$: Choose $\ell\in[N]$ uniformly. Output 
$(q_1,\ldots,q_k)=\big(Q^{(i)}_{\ell,1},\ldots,Q^{(i)}_{\ell,k}\big)$
and ${\sf aux}=\ell$. 
 
\item ${\cal A}(k,j,\bm x,q_j)$: Output $a_j=F_{\bm x}(q_j)$. 
\item ${\cal C}(k,n,a_1,\ldots,a_k,{\sf aux})$:
Compute $y= \sum_{j=1}^k\lambda^{(i)}_{\ell,j} \cdot a_j$ and output  $1_{y=\omega^{(i)}_\ell}$.  
\end{itemize}
    \end{boxedminipage} 
\end{center}
\vspace{-5mm}
\caption{A unified framework for constructing  $(t,k)$-PIR from ${\sf FOASC}(N,k,s,t;\bm \alpha)$}
\label{fig:frwk}
\end{figure}
 
\begin{theorem}
\label{thm:frwk}
If there is an ${\sf FOASC}(N,k,s,t;\bm \alpha)$, where $\bm \alpha$ are $n$ functions from  
 $\mathbb{S}$ to $\mathbb{R}$, then there is 
  a $(t,k)$-PIR protocol $\cal P$  with communication complexity 
${\bf C}_{\cal P}(n,k)=k(\log|\mathbb{S}|+\log|\mathbb{R}|)$.
\end{theorem}

\begin{proof}
It suffices to show that the protocol  $\cal P$ defined by Figure \ref{fig:frwk} is a 
$(t,k)$-PIR with the claimed communication complexity.
For the   correctness of $\cal P$, we have that
\begin{align*}
y&=\sum_{j=1}^k \lambda^{(i)}_{\ell,j} \cdot a_j\\
&=\sum_{j=1}^k \lambda^{(i)}_{\ell,j} \cdot F_{\bm x}(q_j)\\
&= \sum_{j=1}^k \lambda^{(i)}_{\ell,j} \cdot F_{\bm x}(Q^{(i)}_{\ell,j})\\
&=\sum_{j=1}^k \lambda^{(i)}_{\ell,j} \cdot \left(\sum_{\tau=1}^n x_\tau \cdot \alpha_\tau 
\left(Q^{(i)}_{\ell,j}\right)\right)\\
&= \sum_{\tau=1}^n x_\tau \cdot \left(\sum_{j=1}^k \lambda^{(i)}_{\ell,j} \cdot \alpha_\tau 
\left(Q^{(i)}_{\ell,j}\right)\right)\\
&=\bm x\cdot \bm \alpha(\bm Q^{(i)}_{\ell})\cdot \bm \lambda^{(i)}_\ell\\
&=\bm x\cdot   \omega^{(i)}_{\ell} \bm e_n^{(i)}\\
&= \omega^{(i)}_{\ell} x_i.
\end{align*}
Clearly, we have that $1_{y=\omega^{(i)}_\ell}=x_i$ and thus the protocol is correct.

For $t$-privacy, we consider the collusion of any $t$ servers, say 
${\cal S}_{j_1},\ldots,{\cal S}_{j_t}$, and let $T=\{j_1,\ldots,j_t\}$.
As per the querying algorithm $\cal Q$ in  Figure \ref{fig:frwk}, for any $i_1,i_2\in [n]$,
${\cal Q}_T(k,n,i_1)$ (resp. ${\cal Q}_T(k,n,i_2)$) is a random row of the 
$N\times t$ subarray of $\bm Q^{(i_1)}$ (resp. $\bm Q^{(i_2)}$ ) that 
consists of the columns indexed by $T$.
Since $\bm Q^{(i_1)}$ and $\bm Q^{(i_2)}$ are ${\sf OA}(N,k,s,t)$'s,  
${\cal Q}_T(k,n,i_1)$ and ${\cal Q}_T(k,n,i_2)$ are  both  uniformly distributed over 
$\mathbb{S}^t$. Hence, 
 ${\cal Q}_T (k, n, i_1)\xlongequal{\rm id} {\cal Q}_T (k, n, i_2),$ 
i.e., the protocol $\cal P$ is $t$-private.

In our framework, the client sends a query 
$q_j\in \mathbb{S}$ to every server ${\cal S}_j$ and the server  returns an answer 
$a_j\in \mathbb{R}$ to the client.
Thus, the communication complexity   is 
${\bf C}_{\cal P}(n,k)=k(\log|\mathbb{S}|+\log|\mathbb{R}|),$
where $\log|\mathbb{S}|$ (resp. $\log|\mathbb{R}|$) is   the bit length  of every element of
$\mathbb{S}$ (resp. $\mathbb{R}$).
\end{proof}
 
\noindent
{\bf Remark.} Except   Dvir and Gopi \cite{DG15},  
 our framework can capture all PIR protocols considered by this review  with
$\omega^{(i)}_\ell=1$ and thus the reconstructing algorithm can simply output 
$y$. 

\section{PIR Constructions within the Proposed Framework}
\label{sec:pir}

In this section, we show how  several of the most influential   constructions
\cite{CGKS95,WY05,Yek07,Rag07,Efr09,DG15,GKS25} of $(t,k)$-PIR protocols
are captured by the proposed framework, which may   inspire 
new constructions with lower communication complexity. 

\subsection{Protocols based on Covering Codes}
\label{sec:cc}

Chor,   Goldreich,   Kushilevitz, and Sudan \cite{CGKS95} proposed a 
$(1,2)$-PIR  with communication complexity $O(n^{1/3})$ in 1995,
which had been the most influential (1,2)-PIR 
for almost 20 years.

To better understand their protocol,   we   identify every integer $i\in[n]$ with  a tuple $(i_1,i_2,i_3)\in [n^{1/3}]^3$, which   can be done by sorting the tuples in
$[n^{1/3}]^3$   alphabetically and identifying every  $i\in[n]$ with the $i$th tuple. Underlying 
\cite{CGKS95} is a $(1,8)$-PIR   with  communication complexity $O(n^{1/3})$, where
the $n$ bits of the database $\bm x=(x_1,\ldots, x_n)\in \{0,1\}^n$ are organized as a  
 hypercube  of side length $h=n^{1/3}$ (for ease of exposition, assume that 
$n$ is a cubic number), every bit $x_i$ is 
located at a position $(i_1,i_2,i_3)\in [h]^3$ of the hypercube, and
 the 8 servers are named as
${\cal S}_{000},\ldots,{\cal S}_{111}$. 
The $(1,2)$-PIR  is obtained from the  $(1,8)$-PIR  by asking  
${\cal S}_{000}$ to simulate half of the  servers, i.e., 
${\cal S}_{000},{\cal S}_{100},{\cal S}_{010},{\cal S}_{001}$, and
asking ${\cal S}_{111}$ to simulate  the remaining  servers, i.e., 
${\cal S}_{111},{\cal S}_{011},{\cal S}_{101}, {\cal S}_{110}$.
Specifically, the simulation strategy is based on a {\em covering code} with radius 1 for
$\{0,1\}^3$.

\vspace{2mm}
\noindent
{\bf The FOASC and database representation.}
Let  ${\cal H}=\{H_1,\ldots,H_\zeta\}$ be the power set   of  
  $[h]$, where  $\zeta=2^h$.
Denote by $A\oplus B=(A\setminus B)\cup (B\setminus A)$   the symmetric difference of any two sets
$A,B$. 
Within our framework,  underlying the $(1,2)$-PIR   of \cite{CGKS95}  
 is  an ${\sf FOASC}(N,k,s,t;\bm \alpha)$ that consists of $n$ ${\sf OA}(N,k,s,t)$'s $\bm Q^{(1)},\ldots,\bm Q^{(n)}$, where
$N=\zeta^3, k=2, s=\zeta^3, t=1, $ and 
\begin{align}\label{eqn:foacc}
Q^{(i)}_{\ell,1}=(H_{\ell_1},H_{\ell_2},H_{\ell_3});\hspace{2mm}Q^{(i)}_{\ell,2}=(H_{\ell_1}\oplus \{i_1\},
H_{\ell_2}\oplus\{i_2\},H_{\ell_3}\oplus \{i_3\})
\end{align}
for all $i=(i_1,i_2,i_3)\in [h]^3$ and $\ell=(\ell_1,\ell_2,\ell_3) \in[\zeta]^3$. 
 The function  
$F_{\bm x}(\bm z)$ (Eq. (\ref{eq:Fx})) has domain $\mathbb{S}={\cal H}^3$ and range $\mathbb{R}=(\mathbb{F}_2)^{3h+1}$, and 
for all $\tau\in [n]$ and $\bm z=(U,V,W)\in \mathbb{S}$,
\begin{equation}
\begin{split}
\alpha_{\tau}(\bm z)= 1_{\tau\in U\times V\times W}\big\| \big(1_{\tau\in (U\oplus\{c\})\times
  V\times W}\big)_{c\in [h]}
 &\big\| 
\big(1_{\tau\in  U \times (V\oplus\{c\})\times  W}\big)_{c\in [h]} \\
&\big\|\big(1_{\tau\in U\times V\times (W\oplus\{c\})}\big)_{c\in [h]}.
\end{split}
\end{equation}

\vspace{2mm}
\noindent
{\bf The reconstruction coefficients.}
To see that the   FOASC \eqref{eqn:foacc} gives a $(1,2)$-PIR, it suffices to 
note that for all $i=(i_1,i_2,i_3)\in [h]^3$ and $\ell=(\ell_1,\ell_2,\ell_3) \in[\zeta]^3$, 
there is a vector 
\begin{align}
\bm \lambda^{(i)}_\ell=1\big\|\bm e_h^{(i_1)}\big\|\bm e_h^{(i_2)}\big\|\bm e_h^{(i_3)}\big\|1\big\|
\bm e_h^{(i_1)}\big\|\bm e_h^{(i_2)}\big\|\bm e_h^{(i_3)}
\end{align}
such that   $\bm \alpha(\bm Q^{(i)}_\ell)\cdot \bm \lambda^{(i)}_\ell=
\bm e_n^{(i)}$.   
By Theorem \ref{thm:frwk},  the communication complexity of the protocol is
 $${\bf C}(n,k)=2(\log |\mathbb{S}|+\log |\mathbb{R}|)=2(3h+3h+1)=12h+2=O(n^{1/3}).$$

\subsection{Protocols based on Polynomial Interpolations}
\label{sec:pi}

\subsubsection{Lagrange Interpolations}
\label{sec:li}

Chor,   Goldreich,   Kushilevitz, and Sudan \cite{CGKS95} proposed a  Lagrange interpolation based 
$(t,k)$-PIR with communication complexity $O(n^{1/\lfloor (k-1)/t \rfloor})$ in 1995,
which introduced the polynomial interpolation techniques 
to the realm of PIR and 
  inspired many  subsequent  constructions.

\vspace{2mm}
\noindent
{\bf The FOASC and database representation.}
Let $d=\lfloor (k-1)/t \rfloor $ and let $h$ be the least integer such that ${h\choose d}\geq n$. 
Then there exist $n$  vectors ${\bm u}_1,\ldots,{\bm u}_n\in \{0,1\}^h\subseteq \mathbb{F}_p^h$ of Hamming weight $d$, where  $p>k$ is a prime.
Let  ${\bm R}_1,\ldots,{\bm R}_N$ be the $h\times t$ matrices over $\mathbb{F}_p$, where
$N=p^{ht}$.   For 
every  $i\in[n]$ and $\ell\in [N]$, the user's retrieval index $i$ is hidden with a degree $t$  curve
$${\bm q}^{(i)}_\ell(\theta)={\bm u}_i+{\bm R}_\ell \cdot (\theta,\theta^2,\ldots,\theta^t)^\top
.$$
Within our framework, underlying their
  $(t,k)$-PIR is an ${\sf FOASC}(N,k,s,t;\bm \alpha)$ that consists
   of $n$ ${\sf OA}(N,k,s,t)$'s $\bm Q^{(1)},\ldots,\bm Q^{(n)}$, where $N=p^{ht},s=p^h$, and
\begin{align}\label{eqn:foali}
Q^{(i)}_{\ell,j}={\bm q}_\ell^{(i)}(j)
\end{align}
for  all $i\in[n]$, $\ell\in[N]$ and $j\in[k]$. 
 The function  
$F_{\bm x}(\bm z)$ (Eq. (\ref{eq:Fx})) has domain  $\mathbb{S}=\mathbb{F}_p^h$ and range
$\mathbb{R}=\mathbb{F}_p$, where    for all $\tau\in[n]$ and $\bm z\in \mathbb{S}$,
\begin{align}
\alpha_\tau(\bm z)={\bm z}^{\bm u_\tau}.
\end{align}

\vspace{2mm}
\noindent
{\bf The reconstruction coefficients.}
To see  the FOASC  \eqref{eqn:foali} gives a $(t,k)$-PIR, it suffices to note that 
for all $i\in[n]$ and $\ell\in[N]$, there is a vector 
\begin{align}
\bm \lambda_\ell^{(i)}=\left(\prod_{j\in [k]\setminus\{1\}} \frac{j}{j-1},\prod_{j\in [k]\setminus\{2\}} 
\frac{j}{j-2},\
\ldots,\prod_{j\in [k]\setminus\{k\}} \frac{j}{j-k}\right)^\top
\end{align} 
such that   $\bm \alpha(\bm Q^{(i)}_\ell)\cdot \bm \lambda^{(i)}_\ell=
\bm e_n^{(i)}$. 
Specifically, the entries of  $\bm \lambda_\ell^{(i)}$ are 
$k$   coefficients for   Lagrange interpolation.
By Theorem \ref{thm:frwk},  the communication complexity of the protocol is
$${\bf C}(n,k)=k(\log |\mathbb{S}|+\log |\mathbb{R}|)=k(h\log p+\log p)=O(n^{1/\lfloor (k-1)/t \rfloor}).$$

\subsubsection{Hermite Interpolations}
\label{sec:hi}

Woodruff and Yekhanin \cite{WY05} proposed a  Hermite interpolation based 
$(t,k)$-PIR  with communication complexity $O(n^{1/\lfloor (2k-1)/t \rfloor})$ in 2005,
which refined the Lagrange  interpolation techniques of  \cite{CGKS95}
and 
has been the most influential $(t,k)$-PIR  for  $t>1$
during the past 20 years.

\vspace{2mm}
\noindent
{\bf The FOASC and database representation.}
Let $d=\lfloor (2k-1)/t \rfloor $ and let $h$ be the least integer such that ${h\choose d}\geq n$. 
Then there exist $n$  vectors ${\bm u}_1,\ldots,{\bm u}_n\in \{0,1\}^h\subseteq \mathbb{F}_p^h$ of Hamming weight $d$, where  $p>k$ is a prime.  
 Let ${\bm R}_1,\ldots,{\bm R}_N$ be the    $h\times t$ matrices over $\mathbb{F}_p$, where
$N=p^{ht}$.  
For 
every  $i\in[n]$ and $\ell\in [N]$, the user's retrieval index $i$ is hidden with a degree
 $t$   curve 
$${\bm q}^{(i)}_\ell(\theta)={\bm u}_i+{\bm R}_\ell \cdot (\theta,\theta^2,\ldots,\theta^t)^\top
.$$
Within our framework, underlying their
  $(t,k)$-PIR is an ${\sf FOASC}(N,k,s,t;\bm \alpha)$ that consists of $n$ ${\sf OA}(N,k,s,t)$'s
  $\bm Q^{(1)},\ldots,\bm Q^{(n)}$, where $N=p^{ht},s=p^h$, and
\begin{align}\label{eqn:foawy}
Q^{(i)}_{\ell,j}={\bm q}_\ell^{(i)}(j)
\end{align}
for  all $i\in[n]$, $\ell\in[N]$ and $j\in[k]$. 
The function  
$F_{\bm x}(\bm z)$ (Eq. (\ref{eq:Fx})) has domain      $\mathbb{S}=\mathbb{F}_p^h$ and range 
$\mathbb{R}=(\mathbb{F}_p)^{h+1}$,  where   for all $\tau\in[n]$ and $\bm z\in \mathbb{S}$,
\begin{align}\alpha_\tau(\bm z)=\left(
{\bm z}^{\bm u_\tau},\frac{\partial ({\bm z}^{\bm u_\tau})}{\partial z_1},\ldots,\frac{\partial
({\bm z}^{\bm u_\tau})}
{\partial z_h}\right).
\end{align}

\vspace{2mm}
\noindent
{\bf Hermite interpolation basics.}
Based on an observation from \cite{WY05},  for any $k$ distinct  nonzero field elements 
 $\theta_1,\theta_2,\ldots,\theta_k\in \mathbb{F}_p^*$, the $(2k)\times (2k)$ matrix
 $$\bm M_{\theta_1,\theta_2,\ldots,\theta_k}=\left(
       \begin{array}{ccccc}
         1 & \theta_1 & \theta_1^2 & \cdots & \theta_1^{2k-1} \\
                  0 & 1 & 2\theta_1 & \cdots & (2k-1)\theta_1^{2k-2} \\
         1 & \theta_2^1 & \theta_2^2 & \cdots & \theta_2^{2k-1} \\
                  0 & 1 & 2\theta_2 & \cdots & (2k-1)\theta_2^{2k-2} \\
         \vdots & \vdots & \vdots & \cdots & \vdots \\
         1 & \theta_k^1 & \theta_k^2 & \cdots  & \theta_k^{2k-1} \\
          0 & 1 & 2\theta_k & \cdots & (2k-1)\theta_k^{2k-2} \\
       \end{array}
     \right)^\top
 $$
 is nonsingular. Specifically, $\bm M_{ 1,2,\ldots,k}$ is nonsingular and thus
 there is a vector $\bm \mu\in \mathbb{F}_p^{2k}$ such that
 $$ \bm M_{ 1,2,\ldots,k}  \cdot \bm \mu={\bm e}_{2k}^{(1)}.$$
For any degree $<2k$ univariate polynomial $\varphi(\theta)=\varphi_0+\varphi_1\theta+\cdots+
\varphi_{2k-1}\theta^{2k-1}$, we note that 
$$
\big(\varphi(1),\varphi'(1),\ldots,\varphi(k),\varphi'(k)\big)=\big(\varphi_0,\varphi_1,\ldots,\varphi_{2k-1}\big)\cdot 
\bm M_{ 1,2,\ldots,k}.
$$
Therefore, 
one can easily recover   $\varphi_0=\varphi(0)$ from
   $\{\varphi(j), \varphi'(j)\}_{j=1}^k$ as follows
$$\varphi_0=\big(\varphi(1),\varphi'(1),\ldots,\varphi(k),\varphi'(k)\big) 
\cdot \bm \mu.$$

\vspace{2mm}
\noindent
{\bf The reconstruction coefficients.}
For every $\tau\in[n]$,  
 $\phi_\tau(\theta)=({\bm q}_\ell^{(i)}(\theta))^{\bm u_\tau}$ is a polynomial of degree $<2k$. Clearly, there is
a $k(h+1)\times (2k)$ matrix $\bm T_\ell$ that only depends on $\bm R_\ell$   such that  
$$
\big(\phi_\tau(1),\phi_\tau'(1),\ldots,\phi_\tau(k),\phi_\tau'(k)\big) 
= \alpha_\tau\big({\bm Q}_\ell^{(i)}\big) \cdot \bm T_\ell.
$$
Note that $\phi_\tau(0)=({\bm q}_\ell^{(i)}(0))^{\bm u_\tau}=(\bm u_i)^{\bm u_\tau}=1_{\tau=i}$.
Therefore,  for all $i\in[n]$ and $\ell\in[N]$,    we have  
\begin{align}
\bm e_{n}^{(i)}=\big( \phi_1(0), \phi_2(0),\ldots, \phi_n(0)\big)^\top=
\bm \alpha\big({\bm Q}_\ell^{(i)}\big)\cdot \underbrace{\bm 
T_\ell \cdot \bm \mu}_{\bm \lambda^{(i)}_\ell}.
\end{align}
Hence, the    FOASC \eqref{eqn:foawy} gives a $(t,k)$-PIR. 
By Theorem \ref{thm:frwk}, the communication complexity of the protocol is
$${\bf C}(n,k)=k(\log |\mathbb{S}|+\log |\mathbb{R}|)=k(h\log p+(h+1)\log p)=O(n^{1/\lfloor (2k-1)/t \rfloor}).$$

\subsection{Protocols based on Matching Vectors over Finite Fields}
Yekhanin \cite{Yek07} proposed a  
$(1,3)$-PIR  with communication complexity $O(n^{1/r})$
for any integer  $r$ such that $p=2^r-1$ is a Mersenne prime  in 2007.
Assuming that   there are infinitely many Mersenne 
primes, his construction gives a $(1,3)$-PIR  with communication complexity
$O(n^{1/\log\log n})$, which is the {\em first} PIR protocol that uses a {\em constant} number of servers 
and achieves a {\em subpolynomial} communication complexity. 
While the  protocols   in Section \ref{sec:cc} and \ref{sec:pi}  are among
the {\em first} generation of PIR, Yekhanin's  construction 
\cite{Yek07} is best known for   initiating the constructions of the {\em third} generation of 
PIR\footnotemark.\footnotetext{The {\em second} generation of PIR consists of \cite{BIKR02} 
and attracts  limited attention in the realm of PIR.}

\subsubsection{Yekhanin's Construction}
\label{sec:yek1}

The core building block underlying Yekhanin's PIR  is a subset of $\mathbb{F}_p^*$
that is both combinatorially nice and algebraically nice. 
A set $S\subseteq \mathbb{F}_p^*$ is {\em $(h,n)$-combinatorially nice} if there exist two 
sets   $\{{\bm u}_1,\ldots, {\bm u}_n\}, \{
{\bm v}_1,\ldots,{\bm v}_n\}\subseteq \mathbb{F}_p^h$ of vectors such that 
\begin{itemize}
\item
  $\langle{\bm u}_i,{\bm v}_i \rangle=0$ for all $i\in[n]$; and  
  \item $\langle {\bm u}_i,{\bm v}_j  \rangle\in S$ for all $i,j\in[n]$ such that $i\neq j$.
  \end{itemize}
  The two  sets of vectors are said to  form an 
  {\em $S$-matching family} in $\mathbb{F}_p^h$.  
Yekhanin \cite{Yek07} showed that for any Mersenne prime 
 $p=2^r-1$ and any integer $d\geq p-1$, the subgroup 
\begin{align*}  
S=\langle 2 \rangle=\{1,2,\ldots,2^{r-1}\}
\end{align*}
 of  $\mathbb{F}_p^*$ is  $(h,n)$-combinatorially nice for
 $h={d-1+(p-1)/r\choose (p-1)/r}$ and $n={d\choose p-1}.$
Specifically, if we denote by ${\bf 1}_h$ the all-one vector of length $h$,
then the   ${\bm u}_1,\ldots,{\bm u}_n$ constructed by
  \cite{Yek07} satisfy   $\langle {\bm u}_i, {\bf 1}_h \rangle \neq 0$ for all
 $i\in[n]$. 
A set $S\subseteq \mathbb{F}_p^*$ is {\em $k$-algebraically  nice} if there exist two
sets $S_0,S_1\subseteq \mathbb{F}_p$ such that 
\begin{itemize}
\item
  $|S_0|>0, |S_1|=k$, and  
\item  $ |S_0\cap (\sigma+\delta S_1)|\equiv 0 \pmod 2$   for all  $\sigma\in \mathbb{F}_p$
and $\delta\in S$.
\end{itemize}
 Yekhanin \cite{Yek07} showed that 
for any Mersenne prime $p=2^r-1$,  
the set  $S=\langle 2\rangle$  
  is  $k$-algebraically  nice for $k=3$. 
 In particular, if  $g$ is a generator of 
 $\mathbb{F}_{2^r}^*$ and   $\gamma\in \mathbb{F}_p$ is
 an integer such that $1+g+g^\gamma=0$, then  one can choose
 \begin{align*}
 S_1=\{0,1,\gamma\}.
 \end{align*} 
 Furthermore, if 
 $L$ is the linear subspace of $\mathbb{F}_2^p$ that consists of the incidence  vectors of
 the sets  $\{\sigma+\delta S_1\}_{\sigma\in \mathbb{F}_p, \delta\in S}$, 
 then $S_0$ can be any nonempty subset of $\mathbb{F}_p$ whose indicator vector 
belongs to $L^\perp$, the dual space of $L$.  
 From now on, we denote  
 $d_1=0,d_2=1,$ and $d_3=\gamma$.

 \vspace{2mm}
\noindent
{\bf The FOASC and database representation.}
Let      $\mathbb{F}_p^h=\{{\bm w}_1,\ldots,{\bm w}_N\}$, where
$N=p^{h}$. 
Within our framework, underlying Yekhanin's
  $(1,3)$-PIR is an ${\sf FOASC}(N,k,s,t;\bm \alpha)$ that consists of
   $n$ ${\sf OA}(N,k,s,t)$'s $\bm Q^{(1)},\ldots,\bm Q^{(n)}$, where $N=p^h,k=3,s=p^h,t=1$, and
\begin{align}  \label{eqn:foayek}  
Q^{(i)}_{\ell,j}={\bm w}_\ell+d_j \cdot {\bm v}_i
\end{align}
for  all $i\in[n]$,  $\ell\in [N]$ and $j \in[k]$. 
The function  
$F_{\bm x}(\bm z)$ (Eq. (\ref{eq:Fx})) has domain     $\mathbb{S}=\mathbb{F}_p^h$ and range 
$\mathbb{R}= \mathbb{F}_2^{p}$, where for all $\tau\in[n]$ and $\bm z\in \mathbb{S}$
\begin{align}
\alpha_\tau({\bm z})=\Big(1_{\langle {\bm u}_\tau, {\bm z}+\rho\cdot 
{\bf 1}_h\rangle\in S_0}\Big)_{\rho\in \mathbb{F}_p}.
\end{align} 
 
\vspace{2mm}
\noindent
{\bf The reconstruction coefficients.}
 Note that $\langle {\bf u}_\tau, {\bf 1}_h \rangle \neq 0$ for all $\tau\in [n]$. 
There is 
a field element $\rho_\ell\in \mathbb{F}_p$ such that 
$\langle  {\bm u}_i, {\bm w}_\ell+\rho_\ell\cdot {\bf 1}_h\rangle\in S_0$. 
To see that the FOASC \eqref{eqn:foayek} gives a $(1,3)$-PIR, it suffices to note that 
for any $i\in[n]$ and $\ell\in[N]$, there is  a  binary vector 
\begin{align}
\bm \lambda_\ell^{(i)}=\Big(0,\ldots,0,\underbrace{1}_{(\rho_\ell+1){\rm st~entry}},0,\ldots,0,
\underbrace{1}_{(p+\rho_\ell+1){\rm st~entry}},0,\ldots,0,
\underbrace{1}_{(2p+\rho_\ell+1){\rm st~entry}},0,\ldots,0\Big)^\top
\end{align} of length $3p$ and   weight   3 (which is the cardinality of $S_1$)
such that   $\bm \alpha(\bm Q^{(i)}_\ell)\cdot \bm\lambda^{(i)}_\ell=
{\bm e}_n^{(i)}$. 
By Theorem \ref{thm:frwk}, the communication complexity of the protocol is
$${\bf C}(n,k)=3(\log |\mathbb{S}|+\log |\mathbb{R}|)=3(h\log p+p)=O(n^{1/r}).$$

\subsubsection{Raghavendra's Interpretation}

Raghavendra \cite{Rag07} presented a more friendly interpretation of Yekhanin's 
$(1,3)$-PIR, which had  inspired  Efremenko \cite{Efr09},
a milestone in  the third generation  PIR. 

\vspace{2mm}
\noindent
{\bf The FOASC and database representation.}
With the same notation  as in Section \ref{sec:yek1}, Raghavendra \cite{Rag07} considered    a  polynomial  
$$P(\theta)=1+\theta+\theta^\gamma\in \mathbb{F}_{2^r}[\theta]$$ such that
$P(g^\delta)=0$ for   $g\in \mathbb{F}_{2^r}^*$ and all $\delta\in S$,  and  $P(1)=1$. He 
 represented the database $\bm x$ as   
a function 
$F_{\bm x}(\bm z)$ (Eq. (\ref{eq:Fx}))
with domain $\mathbb{S}=\mathbb{F}_p^h$  and range
$\mathbb{R}=\mathbb{F}_{2^r}$, where for all $\tau\in[n]$ and $\bm z\in \mathbb{S}$,
 \begin{align}\label{eqn:Fxrag}
\alpha_\tau({\bm z})=g^{\langle {\bm u}_\tau, {\bm z} \rangle}. 
\end{align}

\vspace{2mm}
\noindent
{\bf The reconstruction coefficients.}
 Note that $P(g^{\langle {\bm u}_\tau, {\bm v}_i \rangle})=1_{\tau=i}$.
To see that the new function $F_{\bm x}$ defined by \eqref{eq:Fx}, \eqref{eqn:Fxrag} and the FOASC defined by \eqref{eqn:foayek} 
   give a $(1,3)$-PIR, it suffices to note that 
for any $i\in[n]$ and $\ell\in[N]$, there is a vector 
\begin{align}
\bm \lambda_\ell^{(i)}=g^{-\langle {\bm u}_i,{\bm w}_\ell \rangle}
\cdot\left(1,1,1\right)^\top
\end{align}
such that   $\bm \alpha(\bm Q^{(i)}_\ell)\cdot \bm \lambda^{(i)}_\ell=
{\bm e}_n^{(i)}$. 
By Theorem \ref{thm:frwk},  the communication complexity of the protocol is
$${\bf C}(n,k)=3(\log |\mathbb{S}|+\log |\mathbb{R}|)=3(h\log p+r)=O(n^{1/r}).$$

\subsection{Protocols based on Matching Vectors over Finite Rings}

For any integer $r\geq 2$, Efremenko \cite{Efr09} proposed a  
$(1,2^r)$-PIR with communication complexity ${\cal L}_r(n)=\exp(O((\log n)^{1/r}(\log \log n)^{1-1/r}))$ 
in 2009. Specifically, 
for $r=2$, their construction can be optimized to give   the first $(1,3)$-PIR with subpolynomial
communication complexity,   without making  any    assumptions such as
the infinity of Mersenne primes  \cite{Yek07}. 
In several subsequent works \cite{IS10,CFLWZ13,DG15,GKS25}, the number  of servers required by 
\cite{Efr09} was further  reduced. 

\subsubsection{Efremenko's Construction}
Let $m=p_1p_2\cdots p_r$ be the product of $r$ distinct   primes
$p_1,p_2,\ldots,p_r$ and let $p$ be a prime/prime power such that
$m|(p-1)$. The {\em canonical set} of $m$ is the  set  
$S_m\subseteq\mathbb{Z}_m$   of  $2^r-1$ nonzero integers 
$\delta\in \mathbb{Z}_m$ that satisfy $\delta \bmod p_j\in \{0,1\}$
for all $j\in[r]$. 
Underlying   \cite{Efr09} is an  
{\em $S_m$-matching family} $\{{\bm u}_1,\ldots,{\bm u}_n\},\{{\bm v}_1,\ldots,{\bm v}_n\}
\subseteq \mathbb{Z}_m^h$ of size 
$n=\exp(O((\log h)^r/(\log\log h)^{r-1}))$ such that
  $\langle{\bm  u}_i,{\bm v}_i \rangle=0$ for all $i\in[n]$; and  
  $\langle {\bm u}_i,{\bm v}_j  \rangle\in S_m$ for all   $i\neq j$.
Such  families can be  obtained from 
Gromulsz's set systems \cite{Gro00}. 
Another ingredient of \cite{Efr09} is an {\em $S_m$-decoding polynomial}
$$P(\theta)=\rho_1\theta^{d_1}+\cdots+\rho_{k} \theta^{d_{k}}\in \mathbb{F}_p[\theta],$$
such that $P(g^\delta)=0$ for all $\delta\in S_m$ and $P(1)=1$, where $g\in\mathbb{F}_p^*$ is of order $m$. 
A trivial construction of $P(\theta)$, i.e., $P(\theta)=\prod_{\delta\in S_m} (\theta-g^\delta)/
\prod_{\delta\in S_m} (1-g^\delta)$, requires $k=2^r$. 

\vspace{2mm}
\noindent
{\bf The FOASC and database representation.}
Let   $\mathbb{Z}_m^h=\{{\bm w}_1,\ldots,{\bm w}_N\}$, where
$N=m^{h}$. 
Within our framework, underlying 
 \cite{Efr09} is 
an ${\sf FOASC}(N,k,s,t;\bm \alpha)$ that consists of
   $n$ ${\sf OA}(N,k,s,t)$'s $\bm Q^{(1)},\ldots,\bm Q^{(n)}$, where $N=m^h,k=2^r,s=m^h,t=1$, and
\begin{align}   \label{eqn:foaefr}
Q^{(i)}_{\ell,j}={\bm w}_\ell+d_j \cdot {\bm v}_i
\end{align} 
for  all $i\in[n]$,  $\ell\in [N]$ and $j \in[k]$. 
The function  
$F_{\bm x}(\bm z)$ (Eq. (\ref{eq:Fx})) has domain 
$\mathbb{S}=\mathbb{Z}_m^h$ and range
$\mathbb{R}=\mathbb{F}_p$, where for all $\tau\in[n]$ and $\bm z\in \mathbb{S}$
\begin{align}
\alpha_\tau({\bm z})=g^{\langle {\bf u}_\tau, {\bm z} \rangle}.
\end{align}

\vspace{2mm}
\noindent
{\bf The reconstruction coefficients.}
Note that $P(g^{\langle {\bm u}_\tau, {\bm v}_i \rangle})=1_{\tau=i}$. 
To see that the  FOASC \eqref{eqn:foaefr} gives a $(1,k)$-PIR, it suffices to note that 
for all $i\in[n]$ and $\ell\in[N]$, there is a vector 
\begin{align}
\bm \lambda_\ell^{(i)}=g^{-\langle {\bm u}_i,{\bm w}_\ell \rangle}\cdot\left(
\rho_1,\rho_2,\ldots,\rho_k\right)^\top
\end{align} 
such that   $\bm \alpha(\bm Q^{(i)}_\ell)\cdot \bm \lambda^{(i)}_\ell=
{\bm e}_n^{(i)}$. 
By Theorem \ref{thm:frwk}, the communication complexity of the protocol is
$${\bf C}(n,k)=k(\log |\mathbb{S}|+\log |\mathbb{R}|)=k(h\log m+\log p)={\cal L}_r(n).$$ 

\subsubsection{Sparse Decoding Polynomials}

Efremenko \cite{Efr09} observed that for $r=2$, a specific modulus such as 
$m=511=7\times 73$ may have an $S_m$-decoding polynomial with 3 monomials and thus give 
a $(1,3)$-PIR rather than a $(1,4)$-PIR. He left it as an open problem to
find $S_m$-decoding polynomials that consist of $<2^r$ monomials
for a general modulus $m=p_1p_2\cdots p_r$.

Shortly after \cite{Efr09}, Itoh and Suzuki \cite{IS10}  showed a composition theorem which states that 
if   $m=p_1\cdots p_r$ and $m'=p'_1\cdots p'_{r'}$ are two coprime moduli 
and there exist an $S_m$-decoding polynomial with $\leq k$ monomials and 
an  $S_{m'}$-decoding polynomial with $\leq k'$ monomials, then there is an
$S_{m''}$-decoding polynomial with $k''\leq kk'$ monomials for $m''=mm'.$
In general, we say that a modulus $m=p_1p_2\cdots p_r$ is  {\em good}  
if it has an $S_m$-decoding polynomial with $<2^r$ monomials.  
The composition theorem implies that a good modulus can help reduce the 
number of required servers in the matching vector based PIR protocols. 

Chee, Feng, Ling, Wang and Zhang \cite{CFLWZ13} conducted an in-depth study of 
Efremenko's open problem and showed that a   surprising result: 
{\em Any Mersenne number (numbers of the form $2^\sigma-1$) that is the product of
two distinct primes  must be a good modulus in Efremenko's construction.}
By computer search, they identified 50 new good modulus of such form,  the least of which is 
$M_{11}=2^{11}-1$ and the largest of which is $M_{7331}=2^{7331}-1$.  
With these good moduli, they obtained $(1,k_r)$-PIR protocols with communication complexity
${\cal L}_{r}(n)$, where $r\geq 2$ and
\begin{align} \label{eqn:kr}
k_r=\left\{
        \begin{array}{ll}
          3^{r/2}, & \hbox{$1<r\leq 103$, $r$ is even;} \\
          8\cdot 3^{(r-3)/2}, & \hbox{$1<r\leq 103$, $r$ is odd;} \\
          (\frac{3}{4})^{51}\cdot 2^r, & \hbox{$r\geq 104$.}
        \end{array}
      \right.
\end{align}
However, it remains an open problem to show that there are infinitely many such Mersenne numbers. 
Further study of the good moduli can be found in \cite{ZLZ22}.

\subsubsection{Hermite-Like Interpolations   over  Exotic Rings}

The transition from  Lagrange interpolation   \cite{CGKS95} to
  Hermite interpolation  \cite{Yek07} allows each server to return 
  more information and thus  halves the number of required servers, in order   to
achieve the same asymptotic communication complexity.  
Inspired by this transition, Dvir and Gopi \cite{DG15}   
  halved the number of servers required by Efremenko \cite{Efr09} 
and obtained  a $(1,2^{r-1})$-PIR with communication complexity 
 ${\cal L}_r(n)$ for any integer $r\geq 2$ in
2015. Specifically, for $r=2$, they got a  $(1,2)$-PIR with communication complexity
 ${\cal L}_2(n)$, which eventually broke  
 the communication complexity record of $O(n^{1/3})$   set by  \cite{CGKS95}.
Their construction was obtained with  Hermite-like interpolations with generalized derivatives over an
exotic ring $\cal R$. 

\vspace{2mm}
\noindent
{\bf The FOASC and database representation.}
 Let $m=p_1p_2\cdots p_r$ be the product of $r$ distinct primes $p_1,p_2,\ldots,p_r$. 
Let $S_m$ be the canonical set of $m$. 
Let $\{{\bm u}_1,\ldots,{\bm u}_n\},\{{\bm v}_1,\ldots,{\bm v}_n\}
\subseteq \mathbb{Z}_m^h$ be an    $S_m$-matching family  of size 
$n=\exp(O((\log h)^r/(\log\log h)^{r-1}))$. 
Let ${\cal R}=\mathbb{Z}_m[g]/(g^m-1)$.
 Let   ${\bm w}_1,\ldots,{\bm w}_N$ be all elements of  $\mathbb{Z}_m^h$, where
$N=m^{h}$. 
  For $k=2^{r-1}$ and every $j\in[k]$, set $d_j=j-1$.
Within our framework, underlying
\cite{DG15} is 
an ${\sf FOASC}(N,k,s,t;\bm \alpha)$ that consists of
   $n$ ${\sf OA}(N,k,s,t)$'s  $\bm Q^{(1)},\ldots,\bm Q^{(n)}$, where $N=m^h,k=2^{r-1},s=m^h,t=1$, and
\begin{align} \label{eqn:foadg}
Q^{(i)}_{\ell,j}={\bm w}_\ell+d_j\cdot {\bm v}_i
\end{align}
for  all $i\in[n]$,  $\ell\in [N]$ and $j \in[k]$. 
The function  
$F_{\bm x}(\bm z)$ (Eq. (\ref{eq:Fx})) has domain 
 $\mathbb{S}=\mathbb{Z}_m^h$ and range
$\mathbb{R}={\cal R}^{h+1}$, where for all $\tau\in[n]$ and $\bm z\in \mathbb{S}$
\begin{align}
\alpha_\tau({\bm z})=
\big(1, {\bm u}_\tau\big) \cdot g^{\langle {\bm u}_\tau, {\bm z} \rangle}.
\end{align}

\vspace{2mm}
\noindent
{\bf Hermite interpolation on multiplicative lines.}
Consider  any polynomial  of the form  $$\varphi(\theta)=\varphi_0+\sum_{\delta\in S_m} \varphi_\delta
 \theta^\delta\in \mathbb{R}[\theta].$$
If we denote  $\bar{\varphi}(\theta)=\sum_{\delta  \in S_m} \delta\cdot  \varphi_\delta \cdot \theta^\delta$, 
 then  there is a $(2k)\times (2k)$ matrix $\bm M$ such that
\begin{align*}
\big( \varphi(g^{d_1}),     \bar{\varphi}(g^{d_1}),
\ldots, \varphi(g^{d_k}),     \bar{\varphi}(g^{d_k})\big) 
= (\varphi_0,\ldots,\varphi_\delta,\ldots)\cdot 
\underbrace{
\left(
  \begin{array}{cccc}
    1 & \cdots & g^{d_1 \delta} & \cdots \\
    0 & \cdots & \delta g^{d_1 \delta} & \cdots \\
    \vdots & \cdots & \vdots & \cdots \\
    1 & \cdots & g^{d_k \delta} & \cdots \\
    0 & \cdots & \delta g^{d_k\delta} & \cdots \\
  \end{array}
\right)^\top}_{\bm M}
\end{align*}
Dvir and Gopi \cite{DG15} showed that there is a vector $\bm \mu \in \mathcal{R}^{2k}$ and a ring element 
$\nu\in \mathcal{R}$ such that
$$
(\nu,0,\ldots,0)^\top=\bm M \cdot \bm \mu
$$
and $\nu \bmod p_j\neq 0$ for all $j\in[r]$. 
Therefore, we have that 
$$\big( \varphi(g^{d_1}),     \bar{\varphi}(g^{d_1}),
\ldots, \varphi(g^{d_k}),     \bar{\varphi}(g^{d_k})\big) \cdot \bm \mu=\varphi_0\nu.$$

\vspace{2mm}
\noindent
{\bf The reconstruction coefficients.}
For every $\tau\in[n]$, consider the univariate polynomial   
$$
\phi_\tau(\theta)=g^{\langle  \bm u_\tau, \bm w_\ell\rangle}\cdot \theta^{\langle \bm u_\tau, \bm v_i \rangle}.$$
Note that the constant term of this function is $g^{\langle  \bm u_i, \bm w_\ell\rangle}\cdot 1_{\tau=i}$. Therefore, 
$$\big( \phi_\tau(g^{d_1}),     \bar{\phi}_\tau(g^{d_1}),
\ldots, \phi_\tau(g^{d_k}),     \bar{\phi}_\tau(g^{d_k})\big) \cdot \bm \mu
=\left\{
   \begin{array}{ll}
     g^{\langle  \bm u_i, \bm w_\ell\rangle} \cdot \nu, & \hbox{$\tau=i$;} \\
     0, & \hbox{otherwise.}
   \end{array}
 \right.
$$
On the other hand, it is not hard to see that there is a $(k(h+1))\times (2k)$ matrix $\bm U$ such that 
\begin{align*}
\big( \phi_\tau(g^{d_1}),     \bar{\phi}_\tau(g^{d_1}),
\ldots, \phi_\tau(g^{d_k}),     \bar{\phi}_\tau(g^{d_k})\big)
= \alpha_\tau(\bm Q^{(i)}_\ell)\cdot \bm U. 
\end{align*}
To see the  FOASC   (\ref{eqn:foadg}) gives a $(1,k)$-PIR, it suffices to note that for all
$i\in[n]$ and $\ell\in [N]$, 
\begin{align}
\underbrace{g^{\langle  \bm u_i, \bm w_\ell\rangle}\nu}_{\omega^{(i)}_\ell} \cdot \bm e_n^{(i)}=
\big(\phi_1(0),\phi_2(0),\ldots,\phi_n(0)\big)^\top\cdot \nu=
\bm \alpha(\bm Q^{(i)}_\ell)\cdot \underbrace{\bm U \bm \mu}_{\bm \lambda^{(i)}_\ell}.
\end{align}
By Theorem \ref{thm:frwk}, the communication complexity of the protocol is
 $${\bf C}(n,k)=k(\log |\mathbb{S}|+\log |\mathbb{R}|)=k(h\log m+(h+1)m\log m)=
 {\cal L}_r(n).$$

\subsubsection{Hermite-Like Interpolations  over Finite Fields}

Recently, 	 Ghasemi,   Kopparty and  Sudan \cite{GKS25} proposed a new
method of combining the Hasse derivatives  with the matching vector based PIR protocols 
 \cite{Yek07,Efr09} and obtained  
 a $(1,\kappa_r)$-PIR protocol with communication complexity
${\cal L}_{r+1}(n)$, where   $\kappa_1=2$ and $\kappa_r=k_r$ for all $r\geq 2$.  
Specifically, for $r=2$, they got a $(1,3)$-PIR    
with communication complexity ${\cal L}_3(n)$, which is  more efficient
than Efremenko \cite{Efr09}, the best $(1,3)$-PIR previously.
Their construction was obtained with  Hermite-like interpolations with Hasse derivatives over a
finite field.

\vspace{2mm}
\noindent
{\bf The decoding problem in Efremenko \cite{Efr09}.}
 Let $m=p_1p_2\cdots p_r$ be the product of $r$ distinct primes $p_1,p_2,\ldots,p_r$. 
Let $p$ be a prime/prime power such that $\gcd(p,m)=1$ and $m|(p-1)$. 
Let
$H_m\subseteq \mathbb{F}_p^*$ be the group of $m$th roots of unity and let $g$ be a generator of
$H_m$.
	 Ghasemi,   Kopparty and  Sudan \cite{GKS25} observed that 
the decoding problem in Efremenko \cite{Efr09}
   is nothing else but the problem of 
interpolating a polynomial of the form
$$\varphi_S(\theta)=\sum_{\delta\in S} \varphi_\delta
 \theta^\delta$$
with evaluations of $\varphi_S(\theta)$ on a  set   $B=\{b_1,\ldots,b_k\}\subseteq H_m$, where
$S$ is a subset of $\mathbb{Z}_m$.

\vspace{2mm}
\noindent
{\bf 0-interpolation set.}
Let $m'=mp$ and let $\phi: \mathbb{Z}_m\times \mathbb{Z}_p\rightarrow \mathbb{Z}_{m'}$
 be the Chinese remainder isomorphism, i.e. $\phi^{-1}(a)=(a \bmod m, a\bmod p)$. 
Suppose that $S\subseteq \mathbb{Z}_m, S'\subseteq \mathbb{Z}_{m'}$ and $e\in \{1,2,\ldots,p\}$ is an integer such that
$S'\subseteq \phi(S\times \{0,1,\ldots,e-1\})$. 
For any multivariate polynomial $F({\bm z})\in \mathbb{F}_p[\bm z]=\mathbb{F}_p[z_1,\ldots,z_h]$
and any nonnegative integer vector ${\bm i}=(i_1,\ldots,i_h)$, 
the ${\bm i}$th {\em Hasse derivative} {$F^{({\bm i})}({\bm z})$} is  the coefficient of 
${\bm y}^{\bm i}$ in the expansion of $F({\bm z}+{\bm y})$, i.e.,
$$
F({\bm z}+{\bm y})=\sum_{{\bm i}} F^{({\bm i})}({\bm z}) {\bm y}^{\bm i}.
$$
For any integer $e\geq 1$, let {$F^{(<e)}({\bm z})$} be the 
vector of   ${\bm i}$th Hasse derivatives $F^{({\bm i})}({\bm z})$ 
for all ${\bm i}=(i_1,\ldots,i_h)$ that satisfies $i_1+\cdots +i_h<e$. 
	 Ghasemi,   Kopparty and  Sudan \cite{GKS25} showed that if 
$B$ is a {\em 0-interpolation set for 
$S$}, then $B$ is  a {\em 0-interpolation set  of multiplicity $e$ for 
$S'$}. In other words, if   $\varphi_S(0)$ is a linear combination of
$\{\varphi_S(b)\}_{b\in B}$, then $\varphi_{S'}(0)$
is a linear combination of $\{\varphi_{S'}^{(<e)}(b)\}_{b\in B}$.
This  critical observation allows them to use a matching family
in $\mathbb{Z}_{m'}^h$ to construct PIR but use 
an $S_{m'}$-decoding polynomial that is as sparse as an $S_m$-decoding polynomial to reconstruct, and thus reduce the number of servers required by the resulting PIR.

\vspace{2mm}
\noindent
{\bf The FOASC and database representation.}
Let  $S_m\subseteq \mathbb{Z}_m, S_{m'}\subseteq \mathbb{Z}_{m'}$ be the canonical sets 
of $m$ and $m'$. Let $\bar{S}_m=S_m\cup\{0\}$ and $\bar{S}_{m'}=S_{m'}\cup \{0\}$.  
Then   $\bar{S}_{m'}=\bar{S}_m\times \{0,1\}$. 
By \cite{GKS25}, 
if  $B=\{b_1,\ldots,b_k\}\subseteq H_m$ is a
 0-interpolation set for $\bar{S}_m$, then $B$ is a 0-interpolation set 
 of multiplicity $e=2$ for $\bar{S}_{m'}$. 
Chee, Feng, Ling, Wang and Zhang \cite{CFLWZ13} and Dvir and Gopi \cite{DG15}
showed that
$\bar{S}_m$ has a 0-interpolation set $B=\{b_1,\ldots,b_k\}\subseteq H_m$ of size $k=\kappa_r$. Therefore,  $\bar{S}_{m'}$ has a 
0-interpolation set (i.e., $B$) of multiplicity $e=2$ and size   $\kappa_r$.
  Let $ \{{\bm u}_1,\ldots,{\bm u}_n\}, \{{\bm v}_1,\ldots,{\bm v}_n\}\subseteq \mathbb{Z}_{m'}^h$ be an $S_{m'}$-matching family of size 
$n=\exp(O((\log h)^{r+1}/(\log\log h)^{r}))$.
Let ${\bm g}_1,\ldots, {\bm g}_N$ be  the elements of $H_m^h$, where
$N=m^h$.  
Within our framework, underlying 
\cite{GKS25} is 
an ${\sf FOASC}(N,k,s,t;\bm \alpha)$ that consists of $n$ ${\sf OA}(N,k,s,t)$'s $\bm Q^{(1)},\ldots,\bm Q^{(n)}$, where $N=m^h,k=\kappa_r,s=m^h,t=1$, and
\begin{align}\label{eqn:foagks}
Q^{(i)}_{\ell,j}={\bm g}_\ell \cdot (b_j)^{{\bm v}_i}
\end{align}
for  all $i\in[n]$,  $\ell\in [N]$ and $j \in[k]$. 
The   function
$
 F_{\bm x}({\bm z})
$
has domain $\mathbb{S}=H_m^h$ and range
$\mathbb{R}=\mathbb{F}_p$ such that for all $\tau\in [n]$ and ${\bm z}\in \mathbb{S}$,
\begin{align}
\alpha_\tau({\bm z})=({\bm z}^{{\bm u}_\tau})^{(<e)} .
\end{align}

\vspace{2mm}
\noindent
{\bf Hermite interpolation on with Hasse derivatives.}
Consider  any polynomial  of the form  $$\varphi(\theta)=\varphi_0+\sum_{\delta\in S_{m'}} \varphi_\delta
 \theta^\delta\in \mathbb{R}[\theta].$$
If we denote  $\bar{\varphi}(\theta)=\varphi^{(1)}(\theta)$, 
 then  
 there is a  column vector    $\bm \mu\in \mathbb{R}^{2k}$ such that
\begin{align*}
\varphi_0=\big( \varphi(b_1),     \bar{\varphi}(b_1),
\ldots, \varphi(b_k),     \bar{\varphi}(b_k)\big)\cdot {\bm \mu}.
\end{align*}

\vspace{2mm}
\noindent
{\bf The reconstruction coefficients.}
For every $\tau\in[n]$, consider the univariate polynomial 
$$
\phi_\tau(\theta)=({\bm g}_\ell)^{\bm u_\tau}\cdot \theta^{\langle \bm u_\tau, \bm v_i\rangle}.$$
Note that the constant term of this function is $(\bm g_\ell)^{ \bm u_i}\cdot 1_{\tau=i}$. Therefore, 
$$\big( \phi_\tau(b_1),     \bar{\phi}_\tau(b_1),
\ldots, \phi_\tau(b_k),     \bar{\phi}_\tau(b_k)\big) \cdot \bm \mu
=\left\{
   \begin{array}{ll}
     (\bm g_\ell)^{ \bm u_i}, & \hbox{$\tau=i$;} \\
     0, & \hbox{otherwise.}
   \end{array}
 \right.
$$
On the other hand, it is not hard to see that there is a $(k(h+1))\times (2k)$ matrix $\bm U$ such that 
\begin{align*}
\big( \phi_\tau(b_1),     \bar{\phi}_\tau(b_1),
\ldots, \phi_\tau(b_k),     \bar{\phi}_\tau(b_k)\big) 
= \alpha_\tau(\bm Q^{(i)}_\ell)\cdot \bm U. 
\end{align*}
To see the  FOASC  \eqref{eqn:foagks} gives a $(1,k)$-PIR, it suffices to note that for all $i\in[n]$
and $\ell\in [N]$,
\begin{align}
\bm e_n^{(i)}=
\big(\phi_1(0),\phi_2(0),\ldots,\phi_n(0)\big)^\top\cdot (\bm g_\ell)^{ -\bm u_i}=
\bm \alpha(\bm Q^{(i)}_\ell)\cdot \underbrace{(\bm g_\ell)^{ \bm -u_i}\cdot \bm U \bm \mu}_{\bm \lambda^{(i)}_\ell}.
\end{align}
By Theorem \ref{thm:frwk}, the communication complexity of the protocol is
 $${\bf C}(n,k)=k(\log |\mathbb{S}|+\log |\mathbb{R}|)=k(h\log m+\log p)=
 {\cal L}_{r+1}(n).$$

\section{Open Problems}
\label{sec:open}

The framework of Section \ref{sec:frwk} gives a unified method of constructing information-theoretic 
$(t,k)$-PIR protocols that can capture the most influential constructions to date. 
Given the state of the art  of IT-PIR, there are several
interesting directions for future research. 
 
\vspace{2mm}
\noindent
{\bf   FOASCs for constructing $(t,k)$-PIR with $t=1$.}
The best known constructions of $(t,k)$-PIR for $t=1$ are due to
Dvir and Gopi \cite{DG15} for $k\leq 26$ and 
due to  Ghasemi, Kopparty and Sudan \cite{GKS25} for all  
$k>26$. These constructions require a composite modulus $m$ and  depend on two critical ingredients:
 the superpolynomial sized $S_m$-matching families from Grolmusz \cite{Gro00}
 and the sparse $S_m$-decoding polynomials from 
Chee, Feng, Ling, Wang and Zhang \cite{CFLWZ13}. 
Given a composite modulus $m=p_1p_2\cdots p_r$ with $r$ prime factors, the communication complexity
of the resulting IT-PIR  can be as low as 
${\cal L}_r(n)$ or ${\cal L}_{r+1}(n)$. However, there is still a big gap between 
the  communication complexity of these protocols and the 
well-known lower bounds \cite{Man98}, which show that  ${\bf C}_{\cal P}(n,k)\geq \Omega(k^2/(k-1)\cdot \log n)$
for any $k$-server IT-PIR. 
New improved constructions of ${\sf FOASC}(N,k,s,1;\bm \alpha)$
may help close the  gaps by giving protocols with lower communication complexity. 
A natural idea of developing better FOASCs includes constructing larger
 $S_m$-matching families \cite{DGY12,CLWZ13,BDL13} or much sparser $S_m$-decoding polynomials. 
It is an interesting open problem to construct
new ${\sf FOASC}(N,k,s,1;\bm \alpha)$ that may result in
$(1,k)$-PIR with   communication complexity   $o({\cal L}_{r}(n))$ for 
$k\leq 26$ and $o({\cal L}_{r+1}(n))$ for 
$k>26 $.

\vspace{2mm}
\noindent
{\bf   FOASCs for constructing $(t,k)$-PIR with $t>1$.}
 The best known constructions of $(t,k)$-PIR for $t>1$ are due to Woodruff and Yekhanin \cite{WY05}
and achieve a communication complexity of  $O(n^{1/\lfloor (2k-1)/t \rfloor})$, which however is much worse
 than the matching vector based $(1,k)$-PIR protocols with subpolynomial communication. 
 Barkol,  Ishai and Weinreb \cite{BIW07} proposed a general transformation
from $(1,k)$-PIR to $(t,k^t)$-PIR that preserves the asymptotic communication complexity.
By  applying this transformation to the matching vector based $(1,k)$-PIR one can obtain 
$(t,k^t)$-PIR with subpolynomial communication for any $t>1$. However, such a transformation results 
in an exponential blowup in the number of required servers. In particular, for a general 
number $k'$ that is not a $t$th power of some integer, it could be very inefficient or even impossible
to use  such a transformation to
construct a $(t,k')$-PIR.  
It is an interesting open problem to construct
new ${\sf FOASC}(N,k,s,t;\bm \alpha)$ that may result in  
$(t,k)$-PIR with   communication complexity   $o(n^{1/\lfloor (2k-1)/t \rfloor})$
for constant $t$ and $k$.

\section{Conclusions}
\label{sec:con}

In this review, we formally define  
families of orthogonal arrays with span capability (FOASC) and 
provide a unified framework for constructing 
multi-server IT-PIR protocols.
We show how to capture the most influential IT-PIR protocols with the proposed framework.
We also put forward several interesting open problems concerning the construction of
FOASCs. 
With the proposed framework, we expect to inspire new FOASCs and thus more efficient
IT-PIR protocols  with 
communication complexity approaching the best known lower bounds.

\section*{Acknowledgements}
This research is partially supported by the National Natural Science Foundation of China (grant No. 
62372299)  and the Open Project Funding of the Key Laboratory of Cyberspace Security Defense  (grant No. 2024-MS-13).
 
\bibliographystyle{plain}

\end{document}